\documentclass{article}

\usepackage[utf8]{inputenc}
\usepackage{tikz}
\usepackage{amsmath}
\usepackage{thm-restate}
\usepackage{mdframed}
\usepackage{amssymb}
\usepackage{amsthm}
\usepackage{thmtools}
\usepackage{relsize}
\usepackage[top=1.3in, bottom=1.2in,left=1.7in,right=1.7in]{geometry}
\usepackage{xcolor}
\usepackage{hyperref}
\usepackage[ruled,vlined]{algorithm2e}
\usepackage[inline]{enumitem}
\usepackage[font=small,margin=.3cm]{caption}

\hypersetup{colorlinks,linkcolor={red!50!black},citecolor={blue!50!black}}

\newcommand\escapetikz[1]{}

\newcommand\opn[1]{\operatorname{#1}}
\newcommand\mc[1]{{\mathcal{#1}}}

\newcommand\emp[1]{\textit{\textbf{#1}}}
\newcommand\nofrac[2]{#1/#2}
\newcommand\q{d}

\newcommand\lift{\opn{lift}}
\newcommand\trim{\text{\textup{\texttt{trim}}}}
\newcommand\Trim{\text{\textup{\texttt{Trim}}}}

\usetikzlibrary{decorations.pathreplacing,decorations.pathmorphing,shapes}
\newcommand\GS{{\mc Z}}

\newcommand\percent[1]{}
\newcommand\supsp{\succeq}
\newcommand\subsp{\preceq}
\newcommand\lft{{\text{L}}}
\newcommand\rgt{{\text{R}}}
\newcommand\sphere{\mc S}

\newcommand\wdth{{n_x}}
\newcommand\hte{{n_y}}

\newcommand\ket[1]{|#1\rangle}
\newcommand\bra[1]{\langle#1|}
\newcommand\bracket[2]{\langle#1|#2\rangle}
\newcommand\Z{\GS}
\newcommand\V{\mc V}
\newcommand\CC{\mathbb{C}}
\newcommand\W{\mc W}
\newcommand\Y{\mc Y}
\newcommand\id{\opn{I}}
\renewcommand\H{\mc H}
\renewcommand\P{P}
\newcommand\AGSP{K}

\newcommand\PAP{\mc K}
\newcommand\B{\mc B}
\newcommand\ot\leftarrow
\newcommand\tr{\opn{tr}}

\newcommand\one{1\kern-.7ex1}
\newcommand\vdim[1]{\opn{dim}(#1)}
\newcommand\Hl{\H_\lft}
\newcommand\Hr{\H_\rgt}
\newcommand\dd{{\tilde d}}

\newcommand\NN{\mathbb N}

\newcommand\Dbound{{D}}
\newcommand\aand{\qquad\text{and}\qquad}
\renewcommand\Bar[1]{\bar{\bar{#1}}}

\hypersetup{colorlinks,linkcolor={red!70!black},citecolor={blue!70!black}}
\providecommand\opn[1]{\operatorname{#1}}
\providecommand\NN{\mathbb N}
\providecommand\mc[1]{\mathcal{#1}}
\renewcommand\H{\mc H}
\providecommand\lft{1}
\providecommand\rgt{2}
\providecommand\Hl{\mc H_{\lft}}
\providecommand\Hr{\mc H_{\rgt}}

\providecommand\iiff{\quad\iff\quad}

\providecommand\lift{\opn{lift}}
\providecommand\liftofrom[2]{\lift_{#1\ot#2}}

\providecommand\id{\mathbf{I}}

\providecommand\other{\text{previous}}
\providecommand\ot{\leftarrow}
\providecommand\into\hookrightarrow

\providecommand\GS{\mc Z}
\providecommand\Z{\GS}

\providecommand\CC{\double{C}}
\providecommand\V{\mc V}

\providecommand\Y{\mc Y}

\providecommand\spec{\opn{spec}}

\providecommand\subsp{\preceq}
\providecommand\supsp{\succeq}
\providecommand\opge\ge
\providecommand\erra{\varphi}
\providecommand\ople\le

\providecommand\plusbounded{\lesssim}
\providecommand\aand{\quad\text{and}\quad}

\providecommand\sphere{\mc S}
\providecommand\vdim[1]{\opn{dim}(#1)}

\providecommand\ONE{1\kern-1.0ex\opn{I}}
\providecommand\shrink{\Delta}
\renewcommand\P{{P}}
\providecommand\nofrac[2]{#1\big/#2}
\providecommand\xpt\backslash
\providecommand\IN[1]{\Gamma_{#1}^\dag}
\providecommand\ON[1]{\Gamma_{#1}}
\providecommand\tofrom[2]{\Pi_{#1\ot\!#2}}
\providecommand\W{\mc W}
\providecommand\B{\mathcal B}

\providecommand\AGSP{{K}}
\providecommand\agsp{{K}}

\providecommand\emp[1]{\textbf{\textit{#1}}}
\providecommand\Bl{\B(\Hl)}
\providecommand\Br{\B(\Hr)}
\providecommand\PAP{\mc K}

\providecommand\bra[1]{\langle#1|}
\providecommand\shann{\opn{S}}

\providecommand\ket[1]{|#1\rangle}
\providecommand\bracket[2]{\langle#1|#2\rangle}
\newtheorem{lemma}{Lemma}
\numberwithin{lemma}{section}
\newtheorem{theorem}[lemma]{Theorem}
\newtheorem*{lemma*}{Lemma}
\newtheorem{proposition}[lemma]{Proposition}
\newtheorem{observation}[lemma]{Observation}
\newtheorem{claim}[lemma]{Claim}
\newtheorem{definition}[lemma]{Definition}
\newtheorem*{definition*}{Definition}
\newtheorem{corollary}[lemma]{Corollary}
\newtheorem{fact}[lemma]{Fact}
\newtheorem{remark}[lemma]{Remark}

\newtheorem*{corollary*}{Corollary}

\newcommand\dgoal{{\delta_{\text{\textup{goal}}}}}
\usepackage[utf8]{inputenc}

\newif\ifblind
\blindfalse

\begin{document}
\title{
Solving degenerate 2D frustration-free spin systems in sub-exponential time
}
\author{
Nilin Abrahamsen
\\
\small Simons Institute for the Theory of Computing,
\\
\small
Berkeley, CA, USA}

\maketitle
\begin{abstract}
We give an improved analysis of approximate ground space projectors (AGSPs) to obtain a sharp error reduction bound in the setting a degenerate ground space. The improved tools imply that the recently proven 2D area law directly extends to sub-exponentially degenerate ground spaces and also yields a sub-exponential-time classical algorithm to compute the ground states. This time complexity cannot be improved beyond sub-exponential, even for the special case of classical unfrustrated systems in 2D. 
\end{abstract}

\section{Introduction}



Computing the ground state energy of a spin system is known to be complete for the computational complexity class QMA in general \cite{kempe_complexity_2006} and even under restrictive assumptions on the interaction geometry, namely when restricting to interactions on the line \cite{hallgren_local_2013}. On the other hand, combining the interaction geometry of a line with the additional assumption of a \emph{spectral gap} leads to a tractable problem \cite{landau_polynomial_2015}. This fact relies on the \emph{area law} \cite{hastings_area_2007} which holds for ground states of gapped spin chains (local Hamiltonians with the geometry of a line). Area laws are the strongest possible \emph{entanglement bounds} for ground states of local Hamiltonians and have so far only been rigorously proven in the case of spin chains, where they coincide with a \emph{constant} entanglement bound \cite{hastings_area_2007}, and in the case of a 2D unfrustrated system \cite{2darea}. In the one-dimensional case the constant entanglement bound further implies an efficient representation of the state as a \emph{matrix product state} (MPS) \cite{vidal_efficient_2003}. It was further shown that an MPS representating a 1D gapped ground state can be computed efficiently \cite{landau_polynomial_2015}. Higher-dimensional analogues of matrix product states are knowns as PEPS, and more generally \emph{tensor networks}. In higher dimensions the property of having an efficient representation as a PEPS state is strictly stronger than satisfying an area law \cite{ge_area_2016}. One can conjecture that efficient PEPS representations should exist for gapped ground states--this can be viewed as a strong version of the area law. 
Still, the existence of a PEPS representation does not imply efficient algorithms, as even computing expectation values for a given 2D PEPS is \#P-complete \cite{Haferkamp_2020}, never mind the optimization problem over the family of PEPS networks.

\emph{Approximate ground space projectors} (AGSPs) are an indispensable tool for proving entanglement bounds on ground states of gapped local Hamiltonians \cite{arad_improved_2012,arad_area_2013,2darea} and for constructing polynomial-time algorithms \cite{landau_polynomial_2015,chubb_computing_2016,arad_rigorous_2017} for gapped spin chains. An AGSP is an operator which shrinks the norms of excited states by a factor $\sqrt\shrink<1$ but leaves the ground states invariant or acts as a dilation on them. 


In results using AGSPs it is often assumed that the ground state be \emph{unique} \cite{arad_improved_2012,arad_area_2013,anshu_entanglement_2020}. Indeed for unique ground states the existence of a $(\shrink,\:R)$-AGSP with $\shrink R\le1/2$ immediately implies an entanglement bound $O(\log R)$ by applying a lemma of Arad, Landau, and Vazirani (\cite{arad_improved_2012} corollary III.4). Here $R$ is the entanglement rank of the AGSP and $\shrink$ the shrinking factor. We call this implication the (non-degenerate case) \emph{off-the-shelf} bound. It reduces the task of proving an area law to that of constructing an AGSP.

Generalizing area laws and algorithms for 1D gapped Hamiltonians from the setting of a unique ground state to a ground space with \emph{degeneracy} (i.e., dimension) $D>1$ has been a focus of several works, starting with the case of a \emph{constant} degeneracy \cite{chubb_computing_2016,huang_area_2014} and later generalized further to \emph{polynomial} degeneracy \cite{arad_rigorous_2017}.  
AGSPs have been used before to prove an area law for polynomially degenerate ground spaces \cite{arad_rigorous_2017} of gapped spin chains, but no direct analogue of the off-the-shelf bound follows from existing tools and analyses. 

\section{Results}
We generalize the off-the-shelf entanglement bound of \cite{arad_improved_2012} to degenerate ground spaces (proposition \ref{mainres}). 
To obtain this generalization we prove the {optimal bound} on how the error of a $\delta$-viable space improves when applying an AGSP (lemma \ref{sharplemma}). 
The improved error reduction bound of lemma \ref{sharplemma} is necessary to prove the off-the-shelf entanglement bound in the degenerate case without assuming a strengthened parameter tradeoff for the AGSP. 


\paragraph{Application to 2D systems}
In a recent breakthrough \cite{2darea} by Anshu, Arad, and Gosset, it was shown that an area law holds for 2D frustration-free spin system satisfying a \emph{uniform (or local) gap condition} \cite{gilyen_preparing_2017}. As remarked in \cite{2darea} citing a preprint of this paper, the off-the-shelf bound proved in our paper immediately extends their result from unique to \emph{sub-exponentially} degenerate ground states.
\\\\
In another paper \cite{abrahamsen2022entanglement} the frustrated version (lemma \ref{frversion}) of our off-the-shelf bound was used to give an area law for 1D bosonic systems.

\subsection{Sub-exponential ground space computation for 2D spin systems}
We then apply our algorithm with a modified version of the AGSP from \cite{2darea}. This allows us to compute the ground states of a uniformly gapped frustration-free 2D spin system in sub-exponential time. This is the first sub-exponential algorithm for computing the ground states of gapped lattice Hamiltonians beyond the one-dimensional setting \cite{landau_polynomial_2015}.  
Assuming the randomized exponential-time typothesis sub-exponential time complexity is the best one can hope for in the present 2D setting, even in the special case of classical constraint-satisfaction problems on a 2D grid.

In appendix \ref{constMPO} we modify the AGSP of \cite{2darea} to arrive at implementable
AGSPs (one for each vertical cut) as in section \ref{algsec} That is, we modify the operator
to bound its entanglement everywhere, yielding an implementable MPO representation. 
Applying algorithm \ref{thealg} to this AGSP yields:

\begin{theorem}\label{maintheorem}
	Let $H$ be a frustration-free Hamiltonian with uniform gap $\gamma$ on the ${\wdth}\times {\hte}$ lattice with $n={\wdth}{\hte}$ qudits. Suppose $\wdth/\hte$ is at most polynomial in $n$. Let $\Dbound$ be a bound on the degeneracy. 
 Then there exists a randomized algorithm with time complexity $D^{O(1)}\exp(n_y^{
1+o(1)})$ which outputs an MPS representing a subspace $\tilde Z\subsp H$ such that $\tilde Z\approx_\delta Z$ where $\delta= 2^{-n_y}$, with probability at least
$1/2$.
\end{theorem}

Without loss of generality we may rotate the lattice such that $\wdth\ge\hte$ and therefore $\hte\le \sqrt n$. The time complexity in theorem \ref{maintheorem} is therefore sub-exponential in $n$, being bounded by
\begin{equation}\label{timecomp}\Dbound^{O(1)}\exp\big[ {n}^{\frac12+o(1)}\big].\end{equation}
The error probability in theorem \ref{maintheorem} is easily reduced by repetition. Indeed, the combined output will be the span of all outputs $\tilde\Z_i$ on which $H$ has energy $0$. 
%
%
%



We verify that the output of theorem \ref{maintheorem} can be used to compute the expectation values of local observables. In fact we may modify the algorithm with a post-processing step which prepares a list of all such expectation values on the ground states.
\begin{corollary}[Post-processing]\label{postproc}
	Let $S\subset(\{\sigma_i\}_{i=1}^{\q^2})^{\otimes n}$ be the set of Pauli observables which act nontrivially on at most $k\le \sqrt n$ spins.
The algorithm of theorem \ref{maintheorem} can be modified to output a 3-dimensional table $T$ such that, for some basis $\{\ket{z_i}\}$ for $\Z$, $|T_{ij}^{\sigma}-\bra{z_i}{\sigma}\ket{z_j}|\le\delta$ for each $\sigma\in S$ and $i,j=1,\ldots,D$ with probability at least $1/2$. The time complexity of the modified algorithm is still \eqref{timecomp}.
\end{corollary}
\begin{proof}
	The modified algorithm runs the algorithm of theorem \ref{maintheorem} and then contracts the resulting MPS to compute each entry of $T$. Contracting the MPS is polynomial in the bond dimension and linear in $n$ \cite{vidal_efficient_2003,paeckel_time-evolution_2019}. Moreover, The number of entries of $T$ is $D^2\binom{n}{k}$, so we can absorb the time complexity in \eqref{timecomp}.
\end{proof}
%

\subsection{Lower bounds}
To state the strongest lower bound we should show hardness in as restrictive a special case as possible. We therefore consider the case when $H$ is a satisfiable classical 3SAT-formula and moreover the degeneracy is $\Dbound=1$, i.e., the satisfying assignment is promised to be unique. Then the local gap is $\gamma=1$, and the satisfying assignment can be found by computing the $1$-local observables using corollary \ref{postproc} to constant accuracy $\delta$. 

\begin{lemma}\label{rect3SAT}
	Let $A$ be the set of \texttt{3SAT} instances on a 2D grid and let $uA\subset A$ be the set of such instances with exactly 1 satisfying assignment. 
	\begin{itemize}
	    \item Suppose there exists a polynomial-time algorithm which given an instance from $uA$ outputs the satisfying assignment with probability 1/2. Then NP equals RP (randomized polynomial time). 
	    \item
	    Suppose there exists a $\exp(n^{o(1)})$-time randomized algorithm which given an instance from $uA$ outputs the satisfying assignment with probability 1/2. Then there exists a $\exp(n^{o(1)})$-time randomized algorithm for for \texttt{SAT} with $n$ variables.
	\end{itemize}
\end{lemma}

\begin{proof}
	\texttt{SAT} is parsimoniously reducible to \texttt{3SAT} \cite{kozen_design_2012} which itself is parsimoniously reducible to \texttt{rectilinear planar 3SAT} \cite{lichtenstein_planar_1982,knuth_problem_1992,demaine_lecture_2014} ({All reductions mentioned are polynomial-time}).  A \texttt{rectilinear planar 3SAT} instance is easily embedded in the 2D grid with 3-local constraints. 
	So there exists a parsimonious reduction $g$ which takes \texttt{SAT} instances to $A$ and \texttt{unique SAT} instances to $uA$. 

	By the Valiant-Vazirani theorem \cite{valiant_np_1985} there exists a randomized reduction $f$ from \texttt{SAT} to \texttt{unique SAT}. 
	Since $g$ preserves uniqueness of solutions $g\circ f$ gives a randomized reduction from \texttt{SAT} to the problem of computing the solution to an instance of $uA$. Since the size $n$ of the $uA$ instance is polynomial in the number of variables $n_0$ of the initial \texttt{SAT} formula, $\exp(n^{o(1)})=\exp(n_0^{o(1)})$.
\end{proof}

It follows that the running time of theorem \ref{maintheorem} and corollary \ref{postproc} cannot be improved to polynomial in $n/\gamma$ unless $NP=RP$. And it cannot be improved to $\exp(n^{o(1)})$ assuming the randomized exponential-time hypothesis.

\subsection{Open problems}


It is interesting whether the sub-exponential \emph{space} complexity of theorem \ref{maintheorem} can be improved further, possibly even so far as to be \emph{polynomial}. This is a question about the existence of efficient representations of ground states, which have so far been elusive beyond gapped spin chains, and which may present additional challenges beyond establishing area laws \cite{ge_area_2016,huang_2d_2020}.   
In a different direction it is interesting whether the sub-exponential \emph{space} complexity of theorem \ref{maintheorem} can be improved further, possibly even so far as to be \emph{polynomial}. This is a question about the existence of efficient representations of ground states, which have so far been elusive beyond gapped spin chains, and which may present additional challenges beyond establishing area laws \cite{ge_area_2016,huang_2d_2020}.   

%

\section{Basic definitions}

Let $\B(\H)$ is the space of linear operators on Hilbert space $\H$ and $\id\in\B(\H)$ the identity. Let $\ople$ denote the Loewner order on operators and write $\Z\subsp\H$ when $\Z$ is a subspace of $\H$.

Given a Hamiltonian $H\in\B(\H)$, an AGSP for $H$ is an operator $\AGSP$ which shrinks the excited states but not the vectors in the ground space $\GS\subsp\H$ of $H$. We do not directly invoke the Hamiltonian itself, as the AGSP property can be captured in terms of just the ground space $\GS$.   Given a subspace $\GS\subsp\H$ let $\P_\GS$ be the projection onto $\GS$. 
The standard definition of an AGSP is the following:
\begin{definition}[\cite{arad_improved_2012,arad_area_2013}]
	A standard $\Delta$-AGSP with target space $\Z\subsp\H$ is an operator $K\in\B(\H)$ which commutes with $\P_\Z$ and such that $K\P_\Z=\P_\Z$ and $\|K\P_{\Z^\perp}\|\le\sqrt\shrink$.
\end{definition}

In the interest of wide applicability we also use a less restrictive definition of an AGSP when stating our error-reduction bound and the subsequent off-the-shelf entanglement bound. This more general AGSP is mainly useful for the frustrated case; in particular it generalizes the \emph{spectral} AGSP of \cite{arad_rigorous_2017} as well as standard AGSPs.
\begin{definition}\label{AGSPdef}
	A general $\shrink$-AGSP with target space $\Z\subsp\H$ is an operator $\AGSP\in\B(\H)$ which commutes with $\P_\Z$ and satisfies
	\begin{enumerate}\item
		$\P_\Z \AGSP^\dag\AGSP\P_\Z\opge\P_\Z$,\qquad i.e., $\AGSP$ is a dilation on $\GS$\label{dilation}
	\item$\|\AGSP\P_{\Z^\perp}\|\le\sqrt\shrink$.\label{shrinkitem}
	\end{enumerate}
\end{definition}
A \emph{spectral} AGSP \cite{arad_rigorous_2017} for a Hamiltonian $H$ corresponds to definition \ref{AGSPdef} with the additional requirement that $\AGSP\opge0$, and that $\AGSP$ and $H$ be simultaneously diagonalizable. 
The condition that $\AGSP$ commute with $\P_\Z$ is equivalent with requiring that $\Z$ and $\Z^\perp$ be closed under $\AGSP$, i.e., following two conditions\footnote{  Indeed, \eqref{toself} imply $\AGSP\P_\GS=\P_\GS\AGSP\P_\GS=\P_\GS\AGSP-\P_\GS\AGSP\P_{\GS^\perp}=\P_\GS\AGSP$, where the last equality is because $\AGSP$ sends $\GS^\perp$ to itself. In the special case where $\AGSP$ is Hermitian it suffices to check one of the implications \eqref{toself}.   } from \cite{arad_improved_2012,arad_area_2013}:
\begin{equation}\label{toself}\ket z\in\GS\implies\AGSP\ket z\in\GS\aand\ket y\in\GS^\perp\implies\AGSP\ket y\in\GS^\perp.\end{equation}


A $(\shrink,R)$-AGSP on a biparitite Hilbert space $\H_L\otimes\H_R$ is a $\shrink$-AGSP $\AGSP\in\Bl\otimes\Br$ with entanglement rank at most $R$, i.e., it is the sum of $R$ tensor products. It will often be useful to consider just the span of the left tensor factors arising in such a decomposition using the following definition:

\begin{definition}
\label{bipagsp}
For a bipartite $\H=\H_L\otimes\H_R$ we say that an operator subspace $\mc K\subsp\B(\H_L)$ is a $\shrink$-AGSP with target space $\Z\subsp\H$ if there exists a $\shrink$-AGSP $K\in\mc K\otimes\B(\H_R)$ with target space $\Z$.
\end{definition}

A $(\shrink,R)$-AGSP on a bipartite Hilbert space gives rise to an AGSP $\PAP\subsp\B(\H_L)$ in the sense of definition \ref{bipagsp} with $\dim\PAP\le R$.

\paragraph{Overlaps and errors}
Let $\sphere(\H)$ be the sphere of unit vectors in $\H$. Let $\Z,\V\subsp\H$ be subspaces.
\begin{definition}
			 $\V$ is $\mu$-overlapping onto $\Z$ (written $\V\supsp_\mu\Z$) if $\|P_\V\ket z\|^2\ge\mu$ for all $\ket z\in\sphere(\Z)$. $\V$ is $\delta$-viable for $\Z$ if $\|P_{\V^\perp}\ket z\|^2\le\delta$ for all $\ket z\in\sphere(\Z)$. 
\end{definition}
$\V$ is $\delta$-viable for $\Z$ iff it is $\mu$-overlapping onto $\Z$ with $\mu=1-\delta$. \emph{The} overlap of $\V$ onto $\GS$ is $\mu=\min_{\ket z\in\sphere(\Z)}\|P_\V\ket z\|^2$ and \emph{the} error of $\V$ onto $\Z$ is $\delta=\max_{\ket z\in\sphere(\Z)}\|P_{\V^\perp}\ket z\|^2$. Two subspaces are $\delta$-close ($\approx_\delta$) if each is $\delta$-viable for the other.

 We say that $\V\subsp\H$ \emph{covers} $\GS\subsp\H$ if $\P_\GS(\V)=\GS$. Equivalently the range of $\P_\GS\P_\V$ is $\Z$, or $P_\Z\P_\V\P_\Z\ge\mu\P_\Z$ for some $\mu>0$. That is, $\V$ covers $\Z$ if its overlap onto $\Z$ is $\mu>0$.
\begin{definition}[Bipartite case]
	Given a target subspace of a bipartite space $\GS\subsp\H_L\otimes\H_R$, a subspace $\mc V\subsp\H_L$ of the left tensor factor is said to be $\delta$-viable for $\GS$ iff $\mc V\otimes\H_R$ is $\delta$-viable for $\GS$.\end{definition}
	
	Typically in the literature the word $\delta$-viable refers exclusively to this bipartite case. But our terminologies are in fact equivalent as one can take $\H_R=\CC$.

\paragraph{Entanglement}
The \emph{von Neumann} entropy $\shann(\rho)$ of a density matrix $\rho$ is the Shannon entropy $\sum_{i}\lambda_i\log(1/\lambda_i)$ of its eigenvalues. For a pure state $\ket\psi\in\H_L\otimes\H_R$ in a bipartite space its entanglement entropy is $\shann(\rho^\psi_L)$ where $\rho^\psi_L=\tr_R(\ket\psi\bra\psi)$ is the reduced density matrix on $\H_L$. This quantity is unchanged if switching the roles of $\H_L$ and $\H_R$.

\section{Improved technical tools}

The entanglement bound and simple algorithm given in this paper both rely on the precise analysis of how overlap is improved when applying an AGSP. This analysis is straightforward when the target is a single vector, but the exact bound was not established previously in the degenerate setting.


\subsection{Error reduction bound}
Consider two subspaces $\GS,\V\subsp\H$ such that $\V$ {covers} $\GS$. 
 Let $\mu>0$ be the overlap of $\V$ onto $\GS$ and let $\delta=1-\mu$, then define the \emp{error ratio} $\erra$ of $\V$ onto $\GS$ as $\varphi=\delta/\mu<\infty$. 
 
 Denoting the largest \emph{principal angle} \cite{galantai_jordans_2006,ben-israel_geometry_1967} between $\GS$ and $\P_\V(\GS)\subsp\V$ as $\theta$ one has that $\mu=\cos^2\theta$ and $\delta=\sin^2\theta$, so we can equivalently write the error ratio of $\V$ onto $\Z$ as
\[\erra=\delta/\mu=\tan^2\theta.\]
\begin{lemma}\label{sharplemma}
Let $\AGSP$ be a general $\shrink$-AGSP for $\GS\subsp\H$, and suppose $\V\subsp\H$ covers $\GS$ with error ratio $\erra$. then $\V':=\AGSP\V=\{\AGSP\ket v:\ket v\in\V\}$ covers $\GS$ and the error ratio $\erra'$ of $\V'$ onto $\GS$ satisfies
\[\erra'\le\shrink\cdot\erra.\]
\end{lemma}
This bound is clearly sharp.\footnote{Consider the $\shrink$-AGSP $\AGSP=\ket0\bra0+\sqrt\shrink\ket1\bra1$ on $\CC^2$ and subspaces $\Z,\V\subsp\CC^2$ spanned by $\ket z=\ket0$ and $\ket v=\frac1{\sqrt{1+\erra}}(\ket 0+\sqrt\erra\ket 1)$.}
Because $\delta'=\frac{\erra'}{1+\erra'}\le\erra'$, lemma \ref{sharplemma} implies:
\begin{corollary}\label{slightly}
If $\V$ is $\delta$-viable for $\GS$ with $\delta=1-\mu<1$, then $\AGSP\V$ is $\delta'$-viable for $\GS$ with error $\delta'\le \shrink\delta/\mu$.
\end{corollary}

The best previous error reduction bound for the general degenerate-case AGSPs (\cite{arad_rigorous_2017} lemma 6) bounded the post-AGSP viability error by
\begin{equation}\label{thelitbound}\delta_{\other}'=\shrink/\mu^2.\end{equation}
The post-AGSP error bound $\delta'$ in corollary \ref{slightly} improves on \eqref{thelitbound} by a factor $\mu\cdot\delta$, 
which is particularly significant when starting in either the small-overlap $\mu\ll1$ or small-error regime $\delta\ll1$. 

The subtlety in proving lemma \ref{sharplemma} comes from the following: While an AGSP is defined in terms of an orthogonal decomposition with respect to the target space $\Z$, the overlap is conversely defined in terms of orthogonal decompositions with respect to the \emph{covering} subspace $\V$. To prove lemma \ref{sharplemma} we replace $\V$ with a subspace $\Y\subsp\V$ and establish a symmetry between $\Y$ and $\Z$. 
 
 In appendix \ref{liftingversion} we also include an alternative proof of lemma \ref{sharplemma} which is more similar in structure to the analysis in \cite{arad_rigorous_2017} lemma 6. In this case we obtain the strengthened bound by improving the `lifting' lemmas (1 and 2) of \cite{arad_rigorous_2017} to have quadratically better dependence on the overlap $\mu$.

%
%


\subsection{Off-the-shelf bound for degenerate ground spaces}
We combine our error reduction bound \ref{sharplemma} with the \emph{bootstrap} \cite{arad_improved_2012} to obtain the off-the-shelf entanglement bound in the degenerate setting.

\begin{proposition}\label{mainres}
	Suppose there exists a general $(\shrink,R)$-AGSP $\AGSP\in\B(\Hl\otimes\Hr)$ such that
	\[R\shrink\le1/2.\]
	Let $\GS$ be the target space of $\AGSP$ and $D=\vdim\GS$ its degeneracy. Then the maximum entanglement entropy of any state $\ket\psi\in\GS$ satisfies the bound
	\[\max_{\ket\psi\in\sphere(\GS)}\shann(\rho^\psi_\lft)=1.01\cdot\log D+O(\log R),\]
	where $\sphere(\GS)$ is the set of unit vectors in $\GS$ and $\shann(\rho^\psi_\lft)$ is the entanglement entropy of $\ket\psi$ between subsystems $\Hl$ and $\Hr$.
\end{proposition}

Proposition \ref{mainres} is proved in section \ref{frsection}.
In the case of a frustrated Hamiltonian the typical AGSP contruction involves spectral \emph{truncations} of parts of the Hamiltonian, incurring an error in the target space of the AGSP. We therefore also prove a version (lemma \ref{frversion}) of proposition \ref{mainres} which is applicable to the frustrated case by allowing the target space to be approximate.

\paragraph{Possible improvements}
Given our formulation of the entanglement bound in proposition \ref{mainres} as a uniform bound over all vectors in $\sphere(\Z)$ it is clear that the bound must include a term $\log D$ corresponding to the degeneracy; consider for example the zero Hamiltonian. On the other hand, since the zero Hamiltonian does not enforce entanglement on its ground states, one may wish to avoid the $\log D$ term at the cost of the bound holding in a weaker sense, say, for a basis. We conjecture that under the conditions of proposition \ref{mainres}, $\Z$ can be written as the span of $D$ vectors $\ket{\psi_i}$ satisfying the entanglement bound $\max_{i=1,\ldots,D}\shann(\rho^\psi_\lft)=O(\log R)$.\footnote{The author thanks Anurag Anshu and David Gosset for a discussion about this non-uniform statement} Even if such an improved entanglement bound holds it seems likely that the uniform bound is the correct notion for algorithms, as it bounds the dimension of a viable space. 


\section{Proof of the off-the-shelf entanglement bound} 

We begin by proving the error reduction bound of lemma \ref{sharplemma}.
\subsection{Proof of improved error reduction bound}

Given Hilbert space $\H$ and subspace $\V\subsp\H$, let $\ON\V:\H\to\V$ denote the orthogonal projection onto $\V$ \emph{ when viewed as a surjective map $\H\to\V$}. Given another subspace $\Z\subsp\H$ we define the \emph{transition map $\tofrom\V\Z$ from $\Z$ to $\V$} as the restriction of $\ON\V$ to domain $\Z$. Formally:
\begin{definition}
Given a subspace $\V\subsp\H$,  let $\ON\V:\H\to\V$ be the adjoint of the inclusion map $\IN\V:\V\into\H$.
The transition map from $\Z$ to $\V$ is $\tofrom\V\Z=\ON\V\IN\Z$.
\end{definition}
The overlap $\mu$ of $\V$ onto $\Z$ equals $\min\spec(\tofrom\Z\V \tofrom\V\Z)$, where $\spec$ is the spectrum.  
The \emph{principal angles between $\Z$ and $\V$} are defined \cite{galantai_jordans_2006,ben-israel_geometry_1967} as the $\arccos$ of the singular values of $\tofrom\V\Z$.
This definition illustrates a symmetry between two subspaces. 


\begin{observation}\label{same}
Let $\V,\Z\subsp\H$ be two subspaces such that each covers the other. Then the overlap of $\V$ onto $\Z$ equals the overlap of $\Z$ onto $\V$.
\end{observation}
\begin{proof}
Let $M=\tofrom\Z\V$. Then $\spec(M M^\dag)\xpt\{0\}=\spec(M^\dag M)\xpt\{0\}$ (Jacobson's lemma). The assumed nonzero overlaps then imply $\spec(M^\dag M)=\spec(MM^\dag)$ and in particular the overlaps agree $\min\spec(MM^\dag)=\min\spec(M^\dag M)$.\end{proof}

If $\V_1\subsp_\mu\V_2$ and $\V_1\supsp_\mu\V_2$ then we say that $\V_1$ and $\V_2$ are mutually $\mu$-overlapping and write $\V_1\parallel_\mu\V_2$.
\begin{corollary}[Symmetry lemma]\label{newswap}
For $\V_1,\V_2\subsp\H$ and $\mu>0$,
\[\V_1\subsp_\mu\V_2\aand\text{$\V_1$ covers $\V_2$}\iiff\V_1\parallel_\mu\V_2.\]
\end{corollary}

\begin{lemma}\label{obv}
For subspaces $\Z,\Y\subsp\H$ and $\mu>0$,
\[\V\supsp_\mu\Z\iiff\P_\V(\Z)\parallel_\mu\Z.\]
\end{lemma}
\begin{proof}Let $\Y=\P_\V(\Z)$.
	($\Leftarrow$) is clear since $\Y\subsp\V$. ($\Rightarrow$): Since $\Y\subsp\V$, $\P_\Y(\Z)=\P_\Y\P_\V(\Z)=\P_\Y(\Y)=\Y$ and thus $\Z$ covers $\Y$. On the other hand $\P_\Z\P_\Y\P_\Z=\P_\Z\P_\V\P_\Z$, so $\V\supsp_\mu\Z$ implies that $\Y\supsp_\mu\Z$. So $\Y\parallel_\mu\Z$ by the symmetry lemma.
\end{proof}


%

\begin{proof}[Proof of lemma \ref{sharplemma}]
Let $\Y=\P_\V\Z$ be the projection of $\Z$ onto the covering space $\V$, and let $\Y'=\AGSP\Y$. By lemma \ref{obv} we have that $\Y\parallel_\mu\Z$ where $\mu=\frac1{1+\varphi}$. Since $\Y'\subsp\AGSP\V$ it suffices to show that $\Y'\parallel_{\mu'}\Z$ where $\mu'=\frac1{1+\shrink\varphi}$. We prove this using the symmetry lemma:

First, $\Y'$ covers $\Z$ because
$P_\Z(\Y')=P_\Z(\AGSP\P_\V\Z)=\AGSP (P_\Z\P_\V\P_\Z)\Z=\AGSP\Z=\Z$. Here we have commuted $\AGSP$ past $\P_\Z$ and used the fact that $\P_\Z\P_\V\P_\Z\opge\mu\P_\Z$ for $\mu>0$ since $\V$ covers $\Z$.

We now compute the overlap of $\Z$ onto $\Y'$. Given $\ket{y'}\in\Y'$ write $\ket{y'}=\AGSP\ket y$ for some $\ket y\in\Y$. Since $\Y\parallel_\mu\Z$ we have 
${\bra{y}\P_{\Z^\perp}\ket{y}}\le\erra{\bra{y}\P_\Z\ket{y}}.$
Apply the AGSP property $\|\AGSP\P_{\Z^\perp}\|\le\sqrt\shrink$ and the dilation property on $\Z$:
\begin{equation}\|\AGSP\P_{\Z^\perp}\ket{y}\|\le\sqrt\shrink\|\P_{\Z^\perp}\ket{y}\|\le\sqrt{\shrink\erra}\|\P_\Z\ket{y}\|\le\sqrt{\shrink\erra}\|\AGSP\P_\Z\ket y\|.\end{equation}
 Recognizing the LHS as $\|\P_{\Z^\perp}\ket{y'}\|$ and the RHS as $\sqrt{\shrink\erra}\|\P_\Z\ket{y'}\|$ establishes that ${\bra{y'}\P_{\Z^\perp}\ket{y'}}\le\shrink\erra{\bra{y'}\P_\Z\ket{y'}}$;
  thus, the error ratio of $\Z$ onto $\Y'$ is at most $\varphi'=\shrink\varphi$. Since we showed that $\Y'$ covers $\Z$, $\varphi'$ is a \emph{mutual} error ratio by the symmetry lemma.
\end{proof}

\subsection{Applying the error reduction bound}\label{EB}

The {bootstrapping argument} \cite{arad_improved_2012,arad_area_2013,arad_rigorous_2017} proves the existence of a subspace $\V\subsp\Hl$ with small dimension and non-negligible overlap with the target space $\GS\subsp\Hl\otimes\Hr$. The argument combines a method for reducing the entanglement of a subspace with one for increasing overlap with the target space (i.e., an AGSP) in such a way that $\vdim\V$ does not increase when concatenating the operations. 

To offset the dimension growth from the AGSP, the entanglement reduction needs to decrease the entanglement by an factor $R$, which means decreasing the overlap by a factor $\Theta(R)$ using the dimension reduction procedure of \cite{arad_rigorous_2017} (appendix \ref{dimred}). One therefore has to apply the $(\shrink,R)$-AGSP in the low-overlap regime $\mu=c/R$. If we used the error bound $\delta'=\shrink/\mu^2$ of \cite{arad_rigorous_2017} then we would need $\shrink<\mu^2=(c/R)^{2}$ to have any bound on the post-AGSP error, hence requiring a bound of the form $R^2\shrink<c^2$ on the parameter tradeoff for the AGSP. In contrast, lemma \ref{sharplemma} weakens this requirement to $\shrink=\mu=cR$. More precisely we will use:
\begin{corollary}\label{sharpapplication}
	Let $\AGSP$ be a $\shrink$-AGSP with target space $\GS\subsp\H$, and suppose $\V\subsp\H$ $\mu$-overlaps onto $\GS$ with $\mu\ge\shrink$. Then $\V'=\AGSP(\V)$ has overlap $\mu'=1/2$ onto $\GS$.
\end{corollary}
\begin{proof}
	$\V$ has error ratio $\erra=\frac{1-\mu}\mu\le\frac1\mu$, so $\V'$ has error ratio $\erra'\le\shrink/\mu\le1$ by lemma \ref{sharplemma}. This corresponds to overlap $\mu'=\frac1{\erra'+1}\ge1/2$.
\end{proof}

The following lemma is proven following the overall argument of \cite{arad_rigorous_2017} proposition 2 and combining it with the sharp error reduction bound in the form of corollary \ref{sharpapplication} to change the condition from a bound on $R^C\shrink$ to one on $R\shrink$. In the following $x\lesssim y$ means $x=O(y\vee1)$ where $\vee$ denotes the maximum.
\begin{lemma}\label{start}
	Let $\GS\subsp\Hl\otimes\Hr$ be a subspace with degeneracy $\vdim\GS=D$. If there exists a $(\shrink,R)$-AGSP $\AGSP\in\B(\Hl)\otimes\B(\Hr)$ with target space $\GS$ and parameters such that
	\begin{equation}\shrink\cdot R\le1/32,\end{equation}
then there exists a left $\frac1{32R}$-overlapping space $\V\subsp\Hl$ onto $\GS$ such that $\vdim\V\plusbounded D\log R$. 
It follows that there exists $\V''$ of dimension $\vdim{\V''}\lesssim DR^2\log R$ which is left $\shrink$-viable for $\GS$.

\end{lemma}
\begin{proof}
	Let $\V$ be a left $\nu=\frac1{32R}$-overlapping space onto $\GS$ whose dimension $V$ is minimal with respect to this property. 
Let $\PAP\subsp\B(\H_L)$ be the $\shrink$-AGSP of dimension $R$ associated to the $(\shrink,R)$-AGSP $\AGSP$, and let $\V'=\PAP\V$ so that $V'=\vdim{\V'}\le RV$. $\shrink\le\nu$ by assumption \eqref{shrinkmu}, so corollary \ref{sharpapplication} yields that $\V'$ is $1/2$-overlapping onto $\GS$. 

By corollary \ref{probmeth} there exists $\Y'\subsp\V'$ which is left $\nu=\frac1{32R}$-overlapping onto $\GS$ and has dimension at most $V/2+O(D\log R\vee\log V)$ since $8V'\cdot\frac{1/(32R)}{1/2}\le V/2$. By minimality of $\V$ we have that $V\le V/2+O(D\log R\vee\log V)$, and rearranging yields the result about $\V$.

The last remark follows by taking $\V''=\PAP^2\V=\PAP\V'$. Then $\V''$ covers $\GS$ with error ratio $\erra''\le\shrink$ by lemma \ref{sharplemma} since $\V'$ has $\erra'=1$, and this upper-bounds the viability error.
\end{proof}

\begin{corollary}\label{existviable}

	Let $\GS\subsp\Hl\otimes\Hr$ be a subspace with degeneracy $\vdim\GS=D$. If there exists a $(\shrink,R)$-AGSP $\AGSP\in\B(\Hl)\otimes\B(\Hr)$ with target space $\GS$ and parameters such that
	\begin{equation}\shrink\cdot R\le1/2,\label{shrinkmu}\end{equation}
then for any $\alpha>0$ there exists a $\alpha$-viable $\V\subsp\H_{[1,\wdth]}$ with $\vdim\V\lesssim\alpha^{-1}{D R^{O(1)}}$.
\end{corollary}
\begin{proof}
	Applying lemma \ref{start} to $\AGSP^5$ there exists a $\mu=\frac1{32R^5}$-overlapping subspace $\V_0$ with dimension $O(\Dbound\log R)$. Let $\V=\AGSP^p\V_0$ where $p={\lceil\log_\Delta(\alpha\mu)\rceil}$. By lemma \ref{sharplemma} the viability error of $\AGSP^p\V_0$ is at most $\Delta^p/\mu\le\alpha$.
	We bound the dimension using $p<1+\log_{R}(\frac1{\alpha\mu})=6+\log_R(32/\alpha)$ which implies $\vdim\V\lesssim\frac{32}{\alpha}R^6\Dbound\log R$.
\end{proof}

\subsection{Subspace overlap $\to$ entanglement of vectors}

The following lemma relates the entanglement of individual ground states to $\delta$-viability. 
\begin{lemma}\label{tailbound}
	Let $\GS\subsp\Hl\otimes\Hr$ and suppose there exists a $\delta$-viable space $\mc V\subset\Hl$ of dimension $V$ for $\GS$. Pick any state $\ket\psi\in\sphere(\GS)$ and write the Schmidt decomposition $\sum_i\sqrt{\lambda_i}\ket{x_i}\ket{y_x}\in\sphere(\GS)$ with non-increasing coefficients. Then we have the tail bound
	\[\sum_{i=V+1}^{\vdim\Hl}\lambda_i\le\sqrt\delta.\]
\end{lemma}
\begin{proof}
	Let $\ket\phi\in\sphere(\GS)$ such that $\bracket{\psi}{\phi}^2\ge1-\delta$, and let $\rho_\psi$ and $\rho_\phi$ be the reduced density matrices on $\Hl$ so that $\lambda_i=\lambda_i^\psi$ are the eigenvalues of $\rho_\psi$. Then, since the trace distance contracts under the partial trace:
	\[\frac12\|\rho_\psi-\rho_\phi\|_1\le\frac12\big\|\ket\psi\bra\psi-\ket\phi\bra\phi\big\|_1=\sqrt{1-\bracket{\psi}{\phi}^2}\le\sqrt\delta,\]
	Let $d\rho=\rho_\psi-\rho_\phi$ and call its non-increasing eigenvalues (not all positive) $\lambda_i^{d\rho}$ and let $\lambda^\phi_i$ be the non-increasing eigenvalues of $\phi$. For $V+1\le i\le\vdim\Hl$, Weyl's inequalities imply $\lambda_{i}^\psi\le\lambda^\phi_{V+1}+\lambda_{i-V}^{d\rho}=\lambda_{i-V}^{d\rho}$.
	Thus $\sum_{i>V}\lambda_{i}\le\sum_j(\lambda_j^{d\rho})_+=\frac12\|d\rho\|_1\le\sqrt\delta$ where $(x)_+=x\vee0$ is the positive part and the middle equality is because $\opn{tr}(d\rho)=0$.
\end{proof}

\subsection{Proof of proposition \ref{mainres}}\label{frsection}

In the case of a frustrated Hamiltonian the AGSP contruction involves a spectral \emph{truncation} of the Hamiltonian on either side of a cut, incurring an error in the target space of the AGSP. We first prove a version of proposition \ref{mainres} which is applicable to the frustrated case by allowing the target space to be approximate. We then specialize to the case of an exact target space to obtain proposition \ref{mainres}.


\begin{lemma}\label{frversion}
	Let $\GS$ with degeneracy $\vdim\GS=D$ be a subspace of bipartite space $\H=\Hl\otimes\Hr$. Let $\tilde\GS_1,\tilde\GS_2,\ldots\subsp\H$ be a sequence of subspaces such that $\tilde\GS_n\approx_{\delta_n}\GS$ where $\delta_1,\delta_2,\ldots$ is a sequence such that $\sum_{n=0}^\infty n\sqrt{\delta_n}=O(1)$.

Let $R\shrink\le1/2$ and suppose there exists a sequence $\AGSP_1,\AGSP_2,\ldots$ such that $\AGSP_n$ is an $(\shrink^n,R^n)$-AGSP for target space $\tilde\Z_n$. Then,
\begin{equation}\max_{\ket\psi\in\sphere(\GS)}\shann(\rho^\psi_\lft)\le(1.01+c_\delta)\log D+O(\log R)\quad\text{where}\quad c_\delta=\sum_{n=1}^\infty\sqrt{\delta_n}.\label{simplified}\end{equation}
\end{lemma}
\begin{proof}
For any $m=5,6\ldots$ we show that $S(\rho^\psi_{\Hl})$ is bounded by
\begin{equation}\label{eqm}(1+\epsilon_m+c_\delta)\log D+O(m\log R),\quad\text{where}\quad\epsilon_m=\frac{\shrink^{m/2}}{1-\shrink^{1/2}}.\end{equation}
\eqref{simplified} then follows by taking $m=17$ since that and $\shrink\le1/2$ yield $\epsilon_n\le0.01$.

Applying lemma \ref{start} to $\PAP_n$ yields a left $\shrink^n$-viable space for $\tilde\GS$ for each $n\ge m$ since $R^n\shrink^n\le1/32$. The lemma implies that $\vdim{\V_n}\lesssim DR^{2n}\log(R^n)$ and hence $\vdim{\V_n}\le CDR^{3n}$ for a constant $C>0$.
	$\V_n$ is $(\shrink^{n/2}+\delta^{n/2})^2$-viable for $\GS$ by the proof of \cite{arad_rigorous_2017} lemma 3.
By lemma \ref{tailbound} the Schmidt coefficients of any state $\ket\psi\in\sphere(\GS)$ satisfy $\sum_{i>C DR^{3n}}\lambda_i\le\shrink^{\frac{n}2}+\sqrt{\delta_n}$ for each $n\ge5$. 

Let $I_0=\{1,2,\ldots,CD\cdot R^{3m}\}$, $I_1,\ldots,I_{m-1}=\emptyset$, and $I_n=\NN\cap(CD\cdot R^{3n},CD\cdot R^{3(n+1)}]$ for $n\ge m$. By the standard decomposition \cite{arad_improved_2012} of the Shannon entropy described in lemma \ref{shannonentropy} (appendix \ref{thestandarddec}),
\begin{align*}\shann(\Lambda_i)&\le\log(C D R^{3m})+\sum_{n=m}^\infty(\shrink^{n/2}+\sqrt{\delta_n})\log(CDR^{3n+3})+\sum_{n=m}^\infty h(\shrink^{n/2}+\sqrt{\delta_n})\\&=(1+\epsilon_m)\log D+O(m\log R)+\sum_{n=m}^\infty h(\shrink^{n/2}+\sqrt{\delta_n}).\end{align*}
We finalize by bounding the rightmost sum. Since $h$ is increasing on $[0,1/e]$, we can bound $h(\shrink^{n/2}+\sqrt{\delta_n})$ by $h(2^{-\frac n2}+\sqrt{\delta_n})$. This in turn is bounded by
\[h(2^{-\frac n2}+\sqrt{\delta_n})\le(2^{-\frac n2}+\sqrt{\delta_n})\log(2^{\frac n2})\le h(2^{-\frac n2})+ n\sqrt{\delta_n}.\]
So $\sum_n h(\shrink^{n/2}+\sqrt{\delta_n})=O(1)$. This establishes \eqref{eqm}.
\end{proof}
The coefficient $1.01$ in lemma \ref{frversion} and proposition \ref{mainres} can be replaced by $1+\epsilon$ for any fixed $\epsilon>0$ by taking $m\propto\log(1/\epsilon)$. The implicit constant of $O(\log R)$ then depends logarithmically on $1/\epsilon$.
\begin{proof}[Proof of proposition \ref{mainres}]
Given AGSP $\AGSP$ with $R\shrink\le1/2$ apply lemma \ref{frversion} to the sequence of AGSPs $\AGSP_n=\AGSP^n$, each with the exact target space $\tilde\GS_n=\GS$ such that we can take $\delta_n=0$. \end{proof}

\section{Simple algorithm given implementable AGSP}

\label{algsec}
\newcommand\algname{\texttt{Groundstates }}

Consider a multipartite Hilbert space $\H=\H_1\otimes\cdots\otimes\H_\wdth$ and a sequence of $\Delta$-AGSP $\AGSP_1,\ldots,\AGSP_\wdth\in\B(\H)$ given as \emph{matrix product operators} (MPOs) with bond dimension $R_{\opn{max}}$. We write $\H_{[1,i]}=\H_1\otimes\cdots\otimes\H_i$. Assume furthermore that for each i, $\AGSP_i$ satisfies a stronger entanglement bound across the i\textsuperscript{th} cut $\H_{[1,i]}|\H_{[i+1,nx]}$, namely the bond dimension of its i\textsuperscript{th} bond is bounded by $R$ such that $R\Delta\le1/2$. If the bond dimension $R_{\opn{max}}$ of the MPO satisfies an appropriate bound, say subexponential, then we call $\AGSP$ an \emph{implementable} AGSP. When applying the algorithm in the 2D case each $\H_i$ will correspond to a column of spins. 

Let $\dd=\max\{\vdim{\H_1},\ldots,\vdim{\H_\wdth}\}$ and let $\Dbound$ be an upper bound on the degeneracy $\vdim\Z$. For each $i=1,\ldots,w$ let $\PAP_{[1,i]}\subsp\B(\H_{[1,i]})$ be the operator subspace encoded by the left part of the MPO for $\AGSP_i$ where the cut bond is left open. Then $\vdim{\PAP_{[1,i]}}\le R$, and $\PAP_{[1,i]}$ is a $\Delta$-AGSP in the sense of definition \ref{bipagsp}. 

Applying our algorithm to the 2D case makes it especially important that the complexity is polynomial in the entanglement rank; this was less essential in the 1D case where the entanglement is constant, and indeed the first algorithms \cite{landau_polynomial_2015} were exponential in the entanglement rank due to the enumeration over \emph{boundary contractions}. To achieve the polynomial dependence on entanglement rank we use the random \emph{sampling} method of \cite{arad_rigorous_2017}.

The algorithm keeps a $\delta$-viable subspace $\Y_{[1,i]}\subsp\H_{[1,i]}=\H_1\otimes\cdots\otimes\H_i$ for the leftmost $i$ sites, similarly to the early algorithm \cite{landau_polynomial_2015} for spin chains. In an iteration it extends with all of the next site $\H_{i+1}$, samples a subspace as in \cite{arad_rigorous_2017}, and applies the left half of the AGSP. Given a viable space $\V$ and an AGSP $\PAP$ as in definition \ref{bipagsp} this means replacing $\V$ with $\V'=\PAP\V=\{L\ket v:L\in\PAP,\ket v\in\V\}$. As with existing algorithms for spin chains \cite{landau_polynomial_2015,chubb_computing_2016,arad_rigorous_2017}, trimming operations are interspersed; we define the trimming procedure $\Y\mapsto\Trim_\varepsilon(\Y)$ in a way that allows a simple self-contained analysis (section \ref{thetrimsec}).

	\begin{algorithm}[h]
	\caption{}
	\label{thealg}
	\KwIn{$\Delta$-AGSPs $\AGSP_1,\ldots,\AGSP_{\wdth}$ given as MPO. Parameters $V\in\NN,\varepsilon,\delta>0$}
	\vspace{1em}
	Set $\Y_{[]}=\CC$\\
	\For{$i=1,\ldots,\wdth$}{
		Sample $\V_{[1,i]}\subsp\Y_{[1,i-1]}\otimes\H_i$ with $\vdim\V_{[1,i]}=V$.
		\\Set $\Y_{[1,i]}=\Trim_\varepsilon(\PAP_{[1,i]}\V_{[1,i]})$.
	}
	\vspace{1em}
	\KwOut{Let $\tilde H=\id-\AGSP^\dag\AGSP$ and $\Y=\Y_{[1,\wdth]}$, and output $\tilde\Z$, the combined eigenspaces of $\tilde H|_{\Y}$ corresponding to eigenvalues $\le\delta$ for $\tilde H$.}
\end{algorithm}
In the last line of algorithm \ref{thealg} we use the notation $A|_\Y=\Gamma A\Gamma^\dag$ where $\Gamma:\H\to\Y$ is the (surjective) projection onto $\Y\subsp\H$ and $\Gamma^\dag:\Y\to\H$ the inclusion map. This line selects the states which are approximately preserved by $\AGSP$, i.e., the target space of the AGSP. 


\begin{proposition}\label{genericalgcor}
	Given an error parameter $0<\dgoal<1$ suppose $R\Delta\le\frac{\dgoal/\dd}{64}$. Let $D$ be the degeneracy of the target space $\Z$ of $\AGSP$. Then there exists a choice of parameters $V,\varepsilon,\delta$ such that with probability at least $1/2$ the output $\tilde Z$ of algorithm \ref{thealg} satisfies $\tilde\Z\approx_{\dgoal}\Z$ and such that the time complexity (and bond dimension of the output) is polynomial in $\Dbound\dd R_{\opn{max}}\wdth/\dgoal$.
\end{proposition}

\subsection{Algorithmic complexity reduction}
\label{thetrimsec}

All existing efficient algorithms for spin chains rely in an essential way on \emph{bond trimming} of matrix product states. We define a trimming procedure which allows for a simple self-contained analysis in the degenerate case. This definition coincides with that of \cite{arad_rigorous_2017} in the {bipartite} (as opposed to multipartite) case: 
\begin{definition}[Bipartite case]
	Given $\Y\subsp\H_{AB}$ and $\varepsilon>0$ introduce the projection $\P_A=\one_{[\varepsilon,\infty)}(\rho^\Y_A)$ where $\rho^\Y_A=\tr_B(\P_\Y)$ is the reduced density matrix and $\one$ denotes an indicator function. Then $\trim_\varepsilon^A(\Y)$ is the image $[\P_A\otimes\id_B](\Y)$.
		\end{definition}
		The trimmed subspace is contained in $\V\otimes\H_B$ where $\V=\P_A(\H_A)$, and Markov's inequality gives the bound \[\vdim\V=\opn{rank}\P_A\le\tr(\rho^\Y_A)/\varepsilon=\vdim\Y/\varepsilon,\]
To define the trimming of $\Y\subsp\H$ in a \emph{multipartite} space $\H=\H_1\otimes\cdots\otimes \H_\wdth$ we simply iterate the bipartite version:
		\begin{definition}
				Given a subspace $\Y\subsp\H_1\otimes\cdots\otimes\H_{j}$, define\[\Trim_\varepsilon(\Y)=\trim_\varepsilon^{1}\circ\trim_\varepsilon^{1,2}\circ\cdots\circ\trim_\varepsilon^{[1,j-1]}(\Y).\]
\end{definition}

\subsection{Analysis of simple trimming procedure}
Since our trimming procedure is just an iteration of the bipartite case its analysis reduces to analyzing the bipartite trimming. We consider a tripartite $\H_{ABC}$ because the subspace $\Y\subsp\H_{AB}$ to be subjected to bipartite trimming is itself viable for a target space $\Z$ on an extended space $\H_{(AB)C}$.
\begin{lemma}\label{bipartitetrimlemma}
	Let $\Z\subsp\H_{ABC}$ and let $\Y\subsp\H_{AB}$ be $\delta$-viable for $\Z$. If there exists $\V\subsp\H_A$ with $\vdim\V=V$ which is $\alpha$-viable for $\Z$, then $\Y_\varepsilon=\trim_\varepsilon^A(\Y)$ is $\delta'$-viable for $\Z$ where $\delta'=\delta+\sqrt{\varepsilon V}+\sqrt{\alpha}$.
\end{lemma}
\begin{proof}
Introduce projectors $\P_+=\one_{[\varepsilon,\infty)}(\rho^\Y_A)$ and $\P_-=\one_{[0,\varepsilon)}(\rho^\Y_A)$ on $\H_A$. Denote extensions of operators and subspaces as $\Bar\P=\P\otimes\id_{BC}$ and $\bar\Y=\Y\otimes\H_C$.

Given any $\ket z\in\sphere(\Z)$ pick $\ket y\in\sphere(\bar\Y)$ satisfying $\bracket zy\ge1-\delta$. Let $\ket{y'}=\Bar\P_+\ket y$ so that $\ket{y'}\in\bar{\Y_\varepsilon}$ and $\|\ket{y'}\|\le1$. Then,
	\begin{equation}\label{decomposition}
	\bracket zy-\bracket{z}{y'}=\bra z\Bar\P_-\ket y=\bra z\Bar\P_\V\Bar\P_-\ket y+\bra z\Bar\P_{\V^\perp}\Bar\P_-\ket y.\end{equation}
	Bound the first term on the RHS by
	\begin{align*}\|\Bar\P_\V\Bar\P_-\ket y\|&=\sqrt{\tr\big(\P_-\P_\V\P_-\tr_{BC}(\ket y\bra y)\big)}
	\le\sqrt{\tr\big(\P_-\P_\V\P_-\rho_A^{\Y}\big)}\le\sqrt{\varepsilon V}.\end{align*}
	since $\|\P_-\rho_A^{\Y}\|\le\varepsilon$ and $\opn{rank}\P_\V=V$. Bound the second term on the RHS of \eqref{decomposition} by $\|\Bar\P_{\V^\perp}\ket z\|\le\sqrt{\alpha}$. By \eqref{decomposition}, $\bracket{z}{y'}\ge\bracket{z}{y}-\sqrt{\varepsilon V}-\sqrt{\alpha}$.
\end{proof}

\begin{corollary}\label{globaltrim}
	Suppose $\Z\subsp\H_{1\ldots j\ldots\wdth}$ is such that for each $i$ there exists a $\alpha$-viable space $\V_{[1,i]}\subsp\H_{[1,i]}$ for $\Z$ with $\vdim{\V_{[1,i]}}\le V$. If $\Y\subsp\H_{[1,j]}$ is $\delta$-viable for $\Z$ then $\Trim_\varepsilon\Y$ is $\delta'$-viable for $\Z$ where $\delta'=\delta+\wdth(\sqrt{\varepsilon V}+\sqrt{\alpha})$.
\end{corollary}

\subsection{Dimension reduction by sampling \cite{arad_rigorous_2017}}

\label{dimred}
Having analyzed the AGSP which achieves the improvement of the overlap we now recall a standard tool for entanglement reduction. 
\begin{lemma}[\cite{arad_rigorous_2017} lemma 5]\label{randlem}
Let $\GS\subsp\Hl\otimes\Hr$ be a subspace with dimension $D$ and let $\W\subsp\Hl$ be left $\mu$-overlapping onto $\GS$ with $\vdim\W=W$. Then a Haar-uniformly random subspace $\V\subsp\W$ of dimension $V\le W$ is left $\nu$-overlapping onto $\GS$ with probability at least $1-\eta$ where
\[\nu=\frac{V}{8W}\cdot\mu\aand\eta=(1+2\nu^{-1/2})^D We^{-V/16}.\]
\end{lemma}
Since $1+2x\le 3x$ for $x>1$ (and in particular for $x=\nu^{-1/2}\ge\sqrt8$) we have the bound on the error probability:
\begin{equation}\eta<(9/\nu)^{D/2}We^{-V/16}.\label{ttbound}\end{equation}
Applying the probabilistic method we obtain:
\begin{corollary}\label{probmeth}
Let $\W\subsp\Hl$ of dimension $W$ be left $\mu$-overlapping onto $\GS\subsp\Hl\otimes\Hr$ with $\vdim\GS=D$. For any $0<\nu\le\mu$ there exists a subspace $\V\subsp\W$ which is left $\nu$-overlapping onto $\GS$ and has dimension
\begin{equation}\label{thisV}V=\Big\lceil 8\Big(W\cdot\frac\nu\mu\vee\big(D\log(9/\nu)+2\log W\big)\Big)\Big\rceil\wedge W.\end{equation}
\end{corollary}
\begin{proof}
	If $V=W$ then $\V=\W$ suffices. Otherwise let $\tilde\nu=\frac{V}{8W}\mu$ be the overlap from lemma \ref{randlem} corresponding to the choice \eqref{thisV} of $V$ and let $\tilde\eta=(9/\nu)^{D/2}We^{-V/16}$. Then $\log \tilde\eta=\tfrac D2\log(9/\nu)+\log W-V/16\le0$ by the choice of $V$.  
	By \eqref{ttbound} the error probability in lemma \ref{randlem} is strictly below $\tilde\eta\le1$ so by the probabilistic method there exists a left $\tilde\nu$-overlapping space. But $\tilde\nu\ge\nu$ which proves the claim.
\end{proof}

%
%
%

\subsection{Analysis of algorithm \ref{thealg}}


\begin{corollary}\label{quanttrim}
Suppose $R\Delta\le1/2$. Given $\delta$ there exists a choice $\varepsilon=\frac1\Dbound(\frac\delta{R\wdth})^{O(1)}$ such that $\trim_\varepsilon$ increases the viability error by at most $\delta$.
\end{corollary}
\begin{proof}
	By corollary \ref{globaltrim} it suffices to verify the existence of $\alpha$-viable subspaces of dimension $V$ such that $\wdth(\sqrt{\varepsilon V}+\sqrt{\alpha})\le\delta$. Let $\alpha=(\frac{\delta}{2\wdth})^2$. By corollary \ref{existviable} we can take $V\lesssim (\nofrac{\wdth}{\delta})^2\Dbound R^{O(1)}$. Then pick $\varepsilon=\frac1V(\frac{\delta}{2\wdth})^2$.
\end{proof}

\begin{lemma}\label{Ylemma}
Given $0<\delta\le1/2$ suppose $R\Delta\le\frac\delta{32\dd}$. Then there exists a choice $V=\Theta(\Dbound\log(R\dd)+\log\wdth)$ and $\varepsilon=\frac1\Dbound(\frac\delta{R\wdth})^{O(1)}$ such that with probability at least $1/2$ each $\Y_{[1,i]}$ is $\delta$-viable for $\Z$ in algorithm \ref{thealg}.
\end{lemma}
\begin{proof}
	At the beginning of the $i$\textsuperscript{th} iteration, $\vdim{\Y_{[1,i-1]}}\le\vdim{\AGSP_{[1,i-1]}\V_{[1,i-1]}}\le RV$. $\Y_{[1,i-1]}\otimes\H_i$ then has dimension at most $\bar Y=\dd RV$.

Let $\nu=\frac1{16\dd R}$. The error probability of lemma \ref{randlem} is bounded by $\eta=(9/\nu)^{\Dbound/2}\bar Ye^{-V/16}$ \eqref{ttbound}.  Pick $V$ such that
\[V-16\log V\ge 8\Dbound\log(144 R\dd)+16\log(2\dd R\wdth).\]
Then $\eta\le(9/\nu)^{\Dbound/2}\bar Ye^{-\frac{\Dbound}2\log(9\cdot16R\dd)-\log(2\dd RV\wdth)}=\frac1{2\wdth}$.  By a union bound lemma \ref{randlem} succeeds at each iteration with probability at least $1/2$.  We perform an induction within this event.\smallbreak\noindent
\textbf{Induction step. }
By the induction hypothesis $\Y_{[1,i-1]}$ is $1/2$-viable for $\Z$.
As lemma \ref{randlem} succeeds, $\V$ has left overlap $\nu=\frac1{16\dd R}$ onto $\Z$. By lemma \ref{sharplemma} $\AGSP_{[1,i]}\V$ is $\delta/2$-viable for $\Z$ since $\Delta/\nu=16\dd R\Delta\le\delta/2$. By corollary \ref{quanttrim} the trimming increases the error only by $\delta/2$, so $\Y_{[1,i]}$ is $\delta$-viable for $\Z$.
\end{proof}

Having shown that $\Y_{[1,\wdth]}$ is $\delta$-viable for $\Z$ it remains to analyze the restriction on the last line of algorithm \ref{thealg}.

\begin{lemma}\label{viableclose}
	If $\Y=\Y_{[1,\wdth]}$ is $\delta$-viable for $\Z$ then the output of algorithm \ref{thealg} is $2\delta$-close to $\Z$. 
\end{lemma}
\begin{proof}
	We show more precisely that $\tilde\Z$ is $\delta/\tilde\gamma$-close to $\Z$ where $\tilde\gamma=1-\Delta\ge1/2$. By the symmetry lemma it suffices to show that \begin{center}
\begin{enumerate*}
	\item\label{ZZ} $\Z$ is ${\delta/\tilde\gamma}$-viable for $\tilde\Z$\aand\vphantom{.1em}
	\item\label{dims}$\vdim{\tilde\Z}\ge\vdim\Z$. 
\end{enumerate*}\end{center}

\ref{ZZ}. By definition $\tilde\Z\subsp\Y$ is such that $H|_{\tilde\Z}\ople\delta$.
Since $\AGSP$ is a $\Delta$-AGSP we can write $\tilde H=\id-\AGSP^\dag\AGSP=0_\Z\oplus\tilde H_{\Z^\perp}$ where $\tilde\gamma\ople\tilde H_{\Z^\perp}$. 
So $\tilde\gamma\P_{\tilde\Z}\P_{\Z^\perp}\P_{\tilde\Z}\ople\P_{\tilde\Z}\tilde H\P_{\tilde\Z}\ople\delta\P_{\tilde\Z}$, which implies that $\Z$ is $\delta/\tilde\gamma$-viable for $\tilde\Z$. 

\ref{dims}.
Since $\Y$ is $\delta$-viable for $\Z$, lemma \ref{obv} implies that $\Z':=\P_\Y\Z\approx_\delta\Z$. Therefore $\P_{\Z'}\tilde H\P_{\Z'}\ople\P_{\Z'}\P_{\Z^\perp}\P_{\Z'}\ople\delta$. So $\Z'$ is a subspace of $\Y$ where $\tilde H$ has energy at most $\delta$ which implies $\vdim{\Z'}\le\vdim{\tilde\Z}$. Item \ref{dims} follows since $\vdim\Z=\vdim{\Z'}$.
\end{proof}

\begin{proof}[Proof of proposition \ref{genericalgcor}]
By lemmas \ref{Ylemma} and \ref{viableclose} we can take $V=\Theta(\Dbound\log(R\dd)+\log\wdth)$, $\varepsilon=\frac1\Dbound(\frac\delta{R\wdth})^{O(1)}$, and $\delta=\dgoal/2$.

Since $\vdim{\AGSP_{[1,i]}\V_{[1,i]}}\le RV$ the bond dimension of the trimmed space $\Y_{[1,i]}$ is bounded by $RV/\varepsilon$ in each iteration. This bounds the bond dimension of $\Y_{[1,i-1]}\otimes\H_{i}$ at the beginning of each iteration by $\dd RV/\varepsilon$, and the same bound holds for the bond dimension of $\V_{[1,i]}$. So the largest bond dimension encountered throughout the algorithm, that of $\AGSP_{[1,i]}\V_{[1,i]}$ before trimming, is bounded by $\dd R_{\opn{max}}2V/\varepsilon=(\Dbound R\wdth/\delta_{\opn{goal}})^{O(1)}\dd\log\dd$. 
The largest subspace dimension is $|Y_{1,i-1}\otimes\H_i| \le RV \dd$. Both are bounded by $(D\dd R_{\opn{max}}\wdth/\delta_{\opn{goal}})^{O(1)}$.
\end{proof}

\section{Constructing an implementable AGSP}\label{constMPO}

\newcommand\Hyr{H_y^{\opn{row}}}
\newcommand\HSr{H_S^{\opn{rows}}}
\newcommand\HS[1]{H_{#1}^{\opn{rows}}}

\newcommand\Hrow\Hyr
\newcommand\Hrows\HSr

\newcommand\ny\hte
\newcommand\nx\wdth

\newcommand\Rmax{R_{\opn{max}}}

The implementable AGSP will be a straightforward modification of the AGSP defined by Anshu, Arad, and Gosset to prove the area law \cite{2darea}
\subsection{Area law AGSP}
For $y = 1,\ldots,\hte$ let $\Hyr$ be a Hamiltonian acting on the horizontal rectangle $[\wdth]\times\{y-1,y\}$ consisting of two adjacent rows. $\Hyr$ will not be the input Hamiltonian $H = \sum_{x\le\wdth,y\le\hte} h_{x,y}$ restricted to these two rows but rather a modified one with small operator norm. For an interval $S\subset[h]$ representing a contiguous set of rows let $\HSr = \sum_{y\in S} \Hyr$. \cite{2darea} defines the following set of operators constructed as products of local Hamiltonian terms.
\begin{definition}[\cite{2darea} definition 3.4]
For $\alpha,\beta > 0$, let $P(\alpha,\beta)$ be the set of products
\[(\HS{S_1})^{r_1}\cdots(\HS{S_k})^{r_k},\qquad r_1 +\cdots+ r_k\le\alpha,\:k\le\beta,\]
where $r_1,\ldots,r_k$ are positive integers and $S_1,\ldots,S_k\subset[\hte]$ are intervals.
\end{definition}
\begin{fact}\label{182}
The AGSP $K_{arealaw}$ constructed in \cite{2darea} is a sum of $\hte^{O(\beta)}$ operators
in $P(\alpha,\beta)$, where $\alpha<\beta = \hte^{1+o(1)}$.
\end{fact}
Suppose each $\Hyr$ satisfies a global\footnote{i.e., for each vertical cut at some horizontal coordinate $x$} bound $b$ on the left-right entanglement rank. Then for any horizontal rectangle described by some set of rows $S$, $HSr$ satisfies a global entanglement rank bound $\hte b$ since $\HSr$ is a sum. In turn, the entanglement
rank of a product in $P(\alpha,\beta)$ is bounded by $(\hte b)^\alpha$, so by fact \ref{182}, $K_{arealaw}$ has a global left-right entanglement rank bound 
\begin{equation}\label{112}
\hte^{O(\beta)}(\hte b)^\alpha=\exp(O(\beta\log\hte)+\alpha\log\hte+\alpha\log b)=\exp[\hte^{1+o(1)}(1+\log b)].
\end{equation}
o(1)+ log C
Let $H = \sum_{x\le\nx,y\le\ny} h_{x,y}$ be the local Hamiltonian. Given a vertical cut of interest,
let W be a small interval of column indices around this cut. In \cite{2darea} the row
operators $\Hrow$ (indexed by row $y$) are taken to be
\begin{equation}\label{113}
    \frac1{2|W|}\big[(I-\Pi_{L,y}) + \sum_{x\in W} h_{x,y} + (I-\Pi_{R,y})\big]
\end{equation}
where $L$ (resp. $R$) indexes the columns to the left (resp. right) of $W$, and $\Pi_{L,y}$ (resp. $\Pi_{R,y}$) is the ground space projector for $\sum_{x\in L} h_{x,y}$ (resp. $\sum_{x\in R} h_{x,y}$). Unfortunately, while this operator has constant entanglement across vertical cuts within W, its Schmidt rank across other vertical cuts can be of order $d^{\Omega(\nx)}$ because $\Pi_{L,y}$ entangles all qudits in row $y$ to the left of $W$. Inserting $b = d^{\Theta(\nx)}$ into (1.12) gives an exponent larger than $n = \ny\nx$, i.e., we get exponentially large entanglement rank for vertical cuts distant from $W$.
\subsection{An implementable variant}
Let $\Lambda_L=\prod_{x\in L,y}\Pi_{x,y}$ and $\Lambda_R=\prod_{x\in R,y}\Pi_{x,y}$ be products of ground space projectors for individual interactions, taken in any order. We modify the AGSP of \cite{2darea} by replacing the horizontal operator of \eqref{113} with the following:
\begin{equation}\label{114}
   \Hyr= \frac14(I-\Lambda_L\Lambda_L^\dag)+\frac1{2|W|}\sum_{x\in W} h_{x,y} + \frac14(I-\Lambda_R\Lambda_R^\dag)
\end{equation}
where the left and right products are taken in any order. $\Hrow$ retains the structure
of \eqref{113} within $W$ necessary to apply the amortization bound in \cite{2darea} at the cut of interest. Furthermore, $\Hrow$ now satisfies a constant entanglement bound $b = O(1)$ across any vertical cut, so \eqref{112} implies that the resulting AGSP has entanglement rank bounded by
\[
R_{\opn{max}} = \exp(\ny^{1+o(1)})
\]
across any vertical cut.

\begin{remark}
The global entanglement rank $\Rmax$ is still larger than the entanglement rank R at the cut of interest since the amortization bound is only applied in one place. In particular we will have $R\shrink\le1/2$ but not the analogous bound for $\Rmax$.
\end{remark}
We verify that substituting the row operator \eqref{114} preserves the shrinking factor of the AGSP $K$. To this end we verify that the row operators $\Hrow$ satisfy the \emph{gap}
and \emph{merge} properties of \cite{2darea}:

\begin{proof}[Proof of gap and merge properties (\cite{2darea} definition 3.2).]
\leavevmode
\begin{itemize}
\item\textbf{Gap:}
The \emph{detectability lemma} \cite{anshu_simple_2016} implies that, if there exists a g-coloring of the interactions, we have the operator ordering
\[I-\Lambda_L\Lambda_L^\dag)\ge\frac{\gamma}{1+g^2}(I-\Pi_L).\]
Here we have used the uniform gap assumption to obtain that $\sum_{x\in L}h_{x,y}$ has gap at least $\gamma$. The analogous bound holds for the right half. It follows that
\begin{equation}\label{115}
\frac14(I-\Lambda_L\Lambda_L^\dag)+\frac14(I-\Lambda_R\Lambda_R^\dag)\ge\frac1{2|W|}[(I-\Pi_L)+(I-\Pi_R)],
\end{equation}
whenever $\frac14\frac\gamma{1+g^2}\ge\frac1{2|W|}$, i.e., when $|W|$ is larger than $2\frac{1+g^2}\gamma = O(1/\gamma)$. Since $\gamma$ is constant and $|W|\to\infty$ as $h\to\infty$ in \cite{2darea}, this eventually holds. Adding $\sum_{x\in W}h_{x,y}$ to both sides of \eqref{115} yields that
\begin{equation}
    \Hrow\ge\eqref{113}
\end{equation}
For a rectangle consisting of a contiguous set $S$ of rows we have the analogous operator lower bound for $\Hrows$. Since the kernel of $\Hrows$ still equals the kernel of $\sum_{y\in S}\sum_xh_{x,y}$, the gap property for \eqref{113} carries over to $\Hrow$.

\item\textbf{Merge:}
The merge property involves only the ground spaces of the operators $\Hrow$. But these are exactly the ground spaces of $\sum_xh_{x,y}$ so this property is not changed.

\end{itemize}
\end{proof}

Finally we note that the trade-off $R\shrink\le1/2$ can be strengthened: Equations (47) and (48) of \cite{2dareafull} (full version on ArXiv) show that
\begin{equation}\label{117}
    D\shrink\le\exp\big[-\ny^{1-O((\log\ny)^{-1/4})+\Omega((\log\ny)^{-1/5})}\big]\le\exp\big[-\ny^{1+\Omega((\log\ny)^{-1/5})}\big].
\end{equation}
Pick any ‘base’ number $B$ growing sufficiently slower than exponentially in
$\ny$. More precisely suppose
\[B = \exp\big[\exp\big[o((\log\ny)^{4/5})\big]\big].\]
Then we can absorb
\[\frac{\log\log B}{\log\ny} =o\big((\log\ny)^{-1/5}\big)\]
if the term $\Omega\big((\log\ny)^{-1/5}\big)$ of \eqref{117}. Thus \eqref{117} implies
\begin{equation}\label{118}
R\shrink\le\exp[-\ny^{1+\frac{\log\log B}{\log\ny}}]=\exp[-\ny e^{\log\log B}]=B^{-\ny}.
\end{equation}
We conclude the proof of theorem \ref{maintheorem} by applying algorithm 1 to the imple-
mentable AGSP $\tilde K$.
\begin{proof}[Proof of theorem \ref{maintheorem}]
Let $\H_i = \H_{\{i\}\times[1,\ny]}$ for $i = 1,\ldots,\nx$ and let $\tilde d \simeq \dim(\H_r) =
d^\ny$. For each position $i$ of the vertical cut the modified $\shrink$-AGSP $K_i$ above is given
by an MPO with bond dimension globally bounded by $\Rmax = \ny^{1+o(1)}$ and bond
 dimension $R < \Rmax$ across the $i$th cut satisfying $R\shrink\le (2d)^{-\ny}=2^{-\ny}/\tilde d$, by \eqref{118}.
 
Apply proposition \ref{genericalgcor} with the virtual qudit dimension $\tilde d = d^\ny$ and error parameter $\delta_{\opn{goal}}=2^{6-\ny}$. The time complexity is $(D\Rmax\nx\tilde d\delta^{-1} )^{O(1)} = D^{O(1)} \exp(\ny^{1+o(1)} )$.
\end{proof}

\ifblind
\else
\section{Acknowledgements}
The author thanks Anurag Anshu and David Gosset for helpful discussions. The author is also grateful to Daniel Grier for pointing out an error in a draft.
\fi


\bibliographystyle{alpha}
\bibliography{bib.bib}

\newcommand{\etalchar}[1]{$^{#1}$}
\begin{thebibliography}{HHEG20}

\bibitem[AAG20]{anshu_entanglement_2020}
Anurag Anshu, Itai Arad, and David Gosset.
\newblock Entanglement {Subvolume} {Law} for 2d {Frustration}-{Free} {Spin}
  {Systems}.
\newblock In {\em Proceedings of the 52nd {Annual} {ACM} {SIGACT} {Symposium}
  on {Theory} of {Computing}}, {STOC} 2020, pages 868--874, New York, NY, USA,
  2020. Association for Computing Machinery.
\newblock event-place: Chicago, IL, USA.

\bibitem[AAG22]{2darea}
Anurag Anshu, Itai Arad, and David Gosset.
\newblock An area law for 2d frustration-free spin systems.
\newblock In {\em Proceedings of the 54th Annual ACM SIGACT Symposium on Theory
  of Computing}, STOC 2022, page 12–18, New York, NY, USA, 2022. Association
  for Computing Machinery.

\bibitem[AAG23]{2dareafull}
Anurag Anshu, Itai Arad, and David Gosset.
\newblock An area law for 2d frustration-free spin systems, 2023.

\bibitem[AAV16]{anshu_simple_2016}
Anurag Anshu, Itai Arad, and Thomas Vidick.
\newblock Simple proof of the detectability lemma and spectral gap
  amplification.
\newblock {\em Physical Review B}, 93(20):205142, 2016.
\newblock Publisher: APS.

\bibitem[AKLV13]{arad_area_2013}
Itai Arad, Alexei Kitaev, Zeph Landau, and Umesh Vazirani.
\newblock An area law and sub-exponential algorithm for {1D} systems.
\newblock {\em arXiv:1301.1162 [cond-mat, physics:quant-ph]}, January 2013.

\bibitem[ALV12]{arad_improved_2012}
Itai Arad, Zeph Landau, and Umesh Vazirani.
\newblock An improved {1D} area law for frustration-free systems.
\newblock {\em Physical Review B}, 85(19):195145, May 2012.

\bibitem[ALVV17]{arad_rigorous_2017}
Itai Arad, Zeph Landau, Umesh Vazirani, and Thomas Vidick.
\newblock Rigorous {RG} algorithms and area laws for low energy eigenstates in
  {1D}.
\newblock {\em Communications in Mathematical Physics}, 356(1):65--105,
  November 2017.

\bibitem[ATB{\etalchar{+}}22]{abrahamsen2022entanglement}
Nilin Abrahamsen, Yu~Tong, Ning Bao, Yuan Su, and Nathan Wiebe.
\newblock Entanglement area law for 1d gauge theories and bosonic systems,
  2022.

\bibitem[BI67]{ben-israel_geometry_1967}
Adi Ben-Israel.
\newblock On the {Geometry} of {Subspaces} in {Euclidean} n-{Spaces}.
\newblock {\em SIAM Journal on Applied Mathematics}, 15(5):1184--1198, 1967.

\bibitem[CF16]{chubb_computing_2016}
Christopher~T. Chubb and Steven~T. Flammia.
\newblock Computing the {Degenerate} {Ground} {Space} of {Gapped} {Spin}
  {Chains} in {Polynomial} {Time}.
\newblock {\em Chicago Journal of Theoretical Computer Science}, 22(1):1--35,
  2016.

\bibitem[Dem14]{demaine_lecture_2014}
Erik Demaine.
\newblock Lecture notes for 6.890., 2014.

\bibitem[GE16]{ge_area_2016}
Yimin Ge and Jens Eisert.
\newblock Area laws and efficient descriptions of quantum many-body states.
\newblock {\em New Journal of Physics}, 18(8):083026, 2016.
\newblock Publisher: IOP Publishing.

\bibitem[GH06]{galantai_jordans_2006}
A.~Gal\'antai and Cs~J. Hegedus.
\newblock Jordan's principal angles in complex vector spaces.
\newblock {\em Numerical Linear Algebra with Applications}, 13(7):589--598,
  2006.

\bibitem[GS17]{gilyen_preparing_2017}
András~Pál Gilyén and Or~Sattath.
\newblock On preparing ground states of gapped hamiltonians: {An} efficient
  quantum {Lovász} local lemma.
\newblock In {\em 2017 {IEEE} 58th {Annual} {Symposium} on {Foundations} of
  {Computer} {Science} ({FOCS})}, pages 439--450. IEEE, 2017.

\bibitem[Has07]{hastings_area_2007}
M.~B. Hastings.
\newblock An {Area} {Law} for {One} {Dimensional} {Quantum} {Systems}.
\newblock {\em Journal of Statistical Mechanics: Theory and Experiment},
  2007(08):P08024--P08024, August 2007.

\bibitem[HHEG20]{Haferkamp_2020}
Jonas Haferkamp, Dominik Hangleiter, Jens Eisert, and Marek Gluza.
\newblock Contracting projected entangled pair states is average-case hard.
\newblock {\em Physical Review Research}, 2(1), jan 2020.

\bibitem[HNN13]{hallgren_local_2013}
Sean Hallgren, Daniel Nagaj, and Sandeep Narayanaswami.
\newblock The local {Hamiltonian} problem on a line with eight states is
  {QMA}-complete.
\newblock {\em arXiv preprint arXiv:1312.1469}, 2013.

\bibitem[Hua14]{huang_area_2014}
Yichen Huang.
\newblock Area law in one dimension: {Degenerate} ground states and {Renyi}
  entanglement entropy.
\newblock {\em arXiv preprint arXiv:1403.0327}, 2014.

\bibitem[Hua20]{huang_2d_2020}
Yichen Huang.
\newblock {2D} {Local} {Hamiltonian} with area laws is {QMA}-complete.
\newblock In {\em 2020 {IEEE} {International} {Symposium} on {Information}
  {Theory} ({ISIT})}, pages 1927--1932. IEEE, 2020.

\bibitem[KKR06]{kempe_complexity_2006}
Julia Kempe, Alexei Kitaev, and Oded Regev.
\newblock The complexity of the local {Hamiltonian} problem.
\newblock {\em SIAM Journal on Computing}, 35(5):1070--1097, 2006.
\newblock Publisher: SIAM.

\bibitem[Koz12]{kozen_design_2012}
Dexter~C Kozen.
\newblock {\em The design and analysis of algorithms}.
\newblock Springer Science \& Business Media, 2012.

\bibitem[KR92]{knuth_problem_1992}
Donald~E. Knuth and Arvind Raghunathan.
\newblock The {Problem} of {Compatible} {Representatives}.
\newblock {\em SIAM J. Discret. Math.}, 5(3):422--427, August 1992.
\newblock Place: USA Publisher: Society for Industrial and Applied Mathematics.

\bibitem[Lic82]{lichtenstein_planar_1982}
David Lichtenstein.
\newblock Planar {Formulae} and {Their} {Uses}.
\newblock {\em SIAM J. Comput.}, 11(2):329--343, 1982.

\bibitem[LVV15]{landau_polynomial_2015}
Zeph Landau, Umesh Vazirani, and Thomas Vidick.
\newblock A polynomial time algorithm for the ground state of one-dimensional
  gapped local {Hamiltonians}.
\newblock {\em Nature Physics}, 11(7):566--569, July 2015.

\bibitem[PKS{\etalchar{+}}19]{paeckel_time-evolution_2019}
Sebastian Paeckel, Thomas Köhler, Andreas Swoboda, Salvatore~R Manmana, Ulrich
  Schollwöck, and Claudius Hubig.
\newblock Time-evolution methods for matrix-product states.
\newblock {\em Annals of Physics}, 411:167998, 2019.
\newblock Publisher: Elsevier.

\bibitem[Vid03]{vidal_efficient_2003}
Guifré Vidal.
\newblock Efficient classical simulation of slightly entangled quantum
  computations.
\newblock {\em Physical review letters}, 91(14):147902, 2003.
\newblock Publisher: APS.

\bibitem[VV85]{valiant_np_1985}
Leslie~G Valiant and Vijay~V Vazirani.
\newblock {NP} is as easy as detecting unique solutions.
\newblock In {\em Proceedings of the seventeenth annual {ACM} symposium on
  {Theory} of computing}, pages 458--463, 1985.

\end{thebibliography}

\appendix

\section{Standard entropy bound \cite{arad_improved_2012} from partial sums}
\label{thestandarddec}
A bound on the Shannon entropy of a probability distribution can be obtained through a dyadic decomposition by following the argument of \cite{arad_improved_2012} lemma III.3. Given a sequence $\Lambda=(\lambda_1,\lambda_2,\ldots)\in[0,1]^{\NN}$ write the Shannon entropy $\shann(\Lambda)=\sum_ih(\lambda_i)$ where $h(x)=x\log(x^{-1})$.
\begin{claim}[\cite{arad_improved_2012,arad_area_2013}]\label{shannonentropy}
Let $\Lambda=(\lambda_1,\lambda_2,\ldots)\in[0,1]^{\NN}$ be a sequence with $\sum_i\lambda_i\le 1$ and write $\Sigma_I=\sum_{i\in I}\lambda_i\le1$ for $I\subset\NN$. Let $I_0,I_1,\ldots$ be a partition on $\NN$ such that $\Sigma_{I_n}\le \gamma_n$ for some sequence of $\gamma_n\in[0,1]$.  If $|I_n|\ge3$ for each $n$, then
\[\shann(\Lambda_i)\le\log|I_0|+\sum_{n=1}^\infty\gamma_n\log(|I_n|)+\sum_{i=1}^\infty h(\gamma_n).\] 
\end{claim}
\begin{proof}
	Since $h$ is concave Jensen's inequality states that for any set of indices $I$, $\frac1{|I|}\sum_{i\in I}h(\lambda_i)\le h(\tfrac1{|I|}\sum_{i\in I}\lambda_i)$. Rearranging yields:
	\begin{equation}\label{Jensens}\sum_{i\in I}h(\lambda_i)\le |I|\cdot h(\nofrac{\Sigma_I}{|I|}).\end{equation}
	$h$ is increasing on $[0,1/e]$, so if $|I|\ge3$ and $\gamma\le1$ is an upper bound on $\Sigma_I$, then $\sum_{i\in I}h(\lambda_i)\le |I|h(\gamma/|I|)=\gamma\log(|I|\gamma^{-1})$. Apply this bound for each $n=1,2,\ldots$. We also have in particular that $\sum_{i\in I}h(\lambda_i)\le\log|I|$. Apply this for $I_0$.
\end{proof}

\section{Alternative proof of sharp error reduction}
\label{liftingversion}


Let $\Z,\V\subsp\H$ be subspaces such that 
$\P_\Z\P_\V\P_\Z\opge\mu\P_\Z$ (i.e, $\V\supsp_\mu\Z$). Lemmas 1 and 2 of \cite{arad_rigorous_2017} state that for every $\ket z\in\Z$ there exists $\ket v\in\V$ with norm at most $\|v\|\le\mu^{-1}\|z\|$ such that $\P_\Z\ket v=\ket z$. The alternative proof of the error reduction lemma \ref{sharplemma} relies on noticing that this statement can be improved quadratically, i.e., we can replace $\mu^{-1}$ with $\mu^{-1/2}$.
\subsection{Quadratically improved lifting lemma}
\begin{definition}
Let $\Z,\V\subsp\H$ be subspaces such that $\V$ covers $\Z$. Define the \emp{lifting operator} from $\Z$ to $\V$ as $\liftofrom\V\Z=\tofrom\V\Z(\ON\Z\P_\V\IN\Z)^{-1}$.
\end{definition}
\begin{lemma}\label{liftlemma}
Given subspaces $\Z,\V\subsp\H$ such that $\V\supsp_\mu\Z$ with $\mu>0$, the lifting operator $\Z\to\V$ satisfies:
\begin{enumerate}
	\item $\P_\Z\circ\liftofrom\V\Z\ket z=\ket z$ for any $\ket z\in\Z$\qquad \hfill\textup{(lifting property)},
		\item \label{opnorm}$\|\liftofrom\V\Z\|\le\mu^{-1/2}$,
	\end{enumerate}
\end{lemma}
\begin{proof}
	The restricted projection $M=\tofrom\Z\V$ is surjective since $\V$ covers $\Z$, so $M^\dag(MM^\dag)^{-1}=\liftofrom\V\Z$ is a well-defined right-inverse\footnote{This right-inverse is a special case of the Moore-Penrose pseudo-inverse, but its role is not analogous to the pseudoinverse in the proof of \cite{arad_rigorous_2017} lemma 6, which was a pseudoinverse of the \emph{AGSP}.} of $M$. This is the lifting property. For the norm bound we write the polar decomposition $\tofrom\Z\V=SV^\dag$ where $S$ is a positive operator on $\Z$ and $V^\dag$ is the adjoint of an isometry $V:\Z\to\V$ (again using that $M$ is surjective). Since $\V\supsp_\mu\Z$ we have $\mu\id_\Z\ople MM^\dag= SV^\dag VS=S^2$ which implies that $S\opge\sqrt\mu\id_\Z$. Then $\|\liftofrom\V\Z\|=\|VS^{-1}\|\le\mu^{-1/2}$.
\end{proof}

We also write $\liftofrom\V\Z$ in the same way when extending its codomain and viewing it as a map $\Z\to\H$.
By Pythagoras' theorem, $\|z\|^2+\|\P_{\Z^\perp}\liftofrom\V\Z\ket z\|^2=\|\liftofrom\V\Z\ket z\|^2\le\mu^{-1}\|z\|^2$.  Since $\erra=\mu^{-1}-1$, rearranging yields:
\begin{corollary}\label{heightcor}
	Let $\Z,\V\subsp\H$ be such that $\V$ covers $\Z$ with error ratio $\erra$. Then for any $\ket z\in\Z$,
	\begin{equation*}\label{perpdec}\liftofrom\V\Z\ket z=\ket z+\P_{\Z^\perp}\liftofrom\V\Z\ket z\qquad\text{where}\qquad\|\P_{\Z^\perp}\liftofrom\V\Z\|\le\sqrt\erra.\end{equation*}
\end{corollary}

The alternative proof of lemma \ref{sharplemma} can now be finalized essentially as in the proof of \cite{arad_rigorous_2017} lemma 6:
\begin{proof}[Finishing the alternative proof of lemma \ref{sharplemma}]
	Write the AGSP as $\agsp_\Z\oplus\agsp_{\Z^\perp}$. 
	Given an arbitrary unit vector $\ket z\in\Z$ pick $\ket{v'}=\AGSP\circ\liftofrom\V\Z\circ\agsp_\Z^{-1}\ket{z}\in\AGSP\V$. It suffices to show that $\bracket{{z}}{v'}^2\ge\mu'\|v'\|^2$ where $\mu'=\frac1{1+\shrink\erra}$: Applying corollary \ref{heightcor} to $\agsp_\Z^{-1}\ket z$ we have the orthogonal decomposition
	\[\ket{v'}=\ket{z}+\ket{h'}\quad\text{where}\quad\ket{h'}=\agsp_{\Z^\perp}(\P_{\Z^\perp}\liftofrom\V\Z)\agsp_\Z^{-1}\ket{z}.\]
where $\|h'\|\le\|\agsp_{\Z^\perp}\|\cdot\|\P_{\Z^\perp}\liftofrom\V\Z\|\le\sqrt{\shrink\erra}$.  
Then $\|v'\|^2\le1+\shrink\erra$ by Pythagoras', so $\bracket{{z}}{v'}^2/\|v'\|^2=1/\|v'\|^2\ge\mu'$.
\end{proof}

\end{document}